\documentclass[11pt,onecolumn]{article}

\usepackage{authblk}

\usepackage{hyperref}
\usepackage{makecell}

\usepackage{hyperref}
\hypersetup{
 colorlinks,
 linkcolor={blue!100!black},
 citecolor={blue!100!black},
 urlcolor={blue!80!black}
}

\usepackage{amsmath,amssymb,amsfonts,amsthm,mathtools,mathdots}

\usepackage{soul}

\newtheorem{theorem}{Theorem}
\newtheorem{remark}{Remark}

\newtheorem{lemma}{Lemma}
\newtheorem{definition}{Definition}

\newtheorem{corollary}{Corollary}

\newcommand{\bF}{\mathbb{F}}

\newcommand{\bN}{\mathbb{N}}

\newcommand{\cC}{\mathcal{C}}

\newcommand{\cH}{\mathcal{H}}

\newcommand{\cN}{\mathcal{N}}

\newcommand{\cT}{\mathcal{T}}

\newcommand{\cX}{\mathcal{X}}

\newcommand{\boldc}{\mathbf{c}}
\newcommand{\boldd}{\mathbf{d}}

\newcommand{\boldg}{\mathbf{g}}
\newcommand{\boldh}{\mathbf{h}}

\newcommand{\boldt}{\mathbf{t}}

\newcommand{\boldv}{\mathbf{v}}

\newcommand{\boldx}{\mathbf{x}}

\newcommand{\boldz}{\mathbf{z}}

\newcommand{\boldmu}{\boldsymbol{\mu}}
\newcommand{\boldnu}{\boldsymbol{\nu}}

\newcommand{\boldomega}{\boldsymbol{\omega}}

\DeclareMathOperator{\poly}{poly}

\DeclareMathOperator{\diag}{diag}
\DeclareMathOperator{\Gab}{Gab}
\DeclareMathOperator{\Tr}{Tr}
\DeclarePairedDelimiter{\floor}{\lfloor}{\rfloor} 

\newcommand{\Span}[1]{{\left\langle {#1} \right\rangle}}
\newcommand{\UH}{\text{UH}}
\newcommand{\LH}{\text{LH}}
\DeclareSymbolFont{bbold}{U}{bbold}{m}{n}
\DeclareSymbolFontAlphabet{\mathbbold}{bbold}
\newcommand{\1}{\mathbbold{1}}
\newcommand{\bal}{\mathrm{bal}}
\newcommand{\pow}{\mathrm{pow}}
\newcommand{\ov}{\overline{\boldv}}
\newcommand{\oom}{\overline{\boldomega}}
\newcommand{\hatom}{\hat{\boldomega}}
\usepackage[table,xcdraw]{xcolor}
\usepackage{chngpage}

\usepackage{xparse}
\DeclarePairedDelimiterX{\set}[1]{\{}{\}}{\setargs{#1}}
\NewDocumentCommand{\setargs}{>{\SplitArgument{1}{;}}m}
{\setargsaux#1}
\NewDocumentCommand{\setargsaux}{mm}
{\IfNoValueTF{#2}{#1} {#1\,\delimsize|\,\mathopen{}#2}}

\renewcommand{\ge}{\geqslant}
\renewcommand{\geq}{\geqslant}
\renewcommand{\le}{\leqslant}
\renewcommand{\leq}{\leqslant}

\title{Hierarchical Erasure Correction of Linear Codes}
\author[1]{\textbf{Netanel Raviv}}
\author[2]{\textbf{Moshe Schwartz}}
\author[3]{\textbf{Rami Cohen}}
\author[3]{\textbf{Yuval Cassuto}}
\affil[1]{Department of Computer Science and Engineering, Washington University in St. Louis, St. Louis, MO}
\affil[2]{School of Electrical and Computer Engineering, Ben-Gurion University of the Negev, Beer-Sheva, Israel}
\affil[3]{Faculty of Electrical Engineering, Technion--Israel Institute of Technology, Haifa, Israel}

	\date{}
\begin{document}

	\maketitle
	\thispagestyle{empty}




\begin{abstract}
Linear codes over finite extension fields have widespread applications in theory and practice. In some scenarios, the decoder has a sequential access to the codeword symbols, giving rise to a hierarchical erasure structure. In this paper we develop a mathematical framework for hierarchical erasures over extension fields, provide several bounds and constructions, and discuss potential applications in distributed storage and flash memories. Our results show intimate connection to Universally Decodable Matrices, as well as to Reed-Solomon and Gabidulin codes. 
\end{abstract}




\section{Introduction}

For a prime power~$q$, let~$\bF_q$ be the finite field with~$q$ elements. For a positive integer~$\alpha$, let~$\bF_{q^\alpha}$ be its algebraic extension of degree~$\alpha$, that can be viewed as a vector space of dimension~$\alpha$ over~$\bF_q$ by fixing an ordered basis~$\boldomega=(\omega_1,\ldots,\omega_\alpha)$ of~$\bF_{q^\alpha}$ over~$\bF_q$.
For an integer~$n$, a code~$\cC \subseteq\bF_{q^\alpha}^n$ is called \emph{linear over~$\bF_{q^\alpha}$} (or linear, in short) if it is a linear subspace of~$\bF_{q^\alpha}^n$, in which case its dimension is denoted by~$k$. 

Traditionally, the coding-theoretic literature discusses encoding and decoding of linear codes under \emph{erasures}, i.e., where codeword symbols are replaced by some symbol~$*$ outside the field, and \emph{errors}, where codeword symbols are replaced by arbitrary field elements. The mathematical framework for erasures and errors is very well understood, and bounds and matching constructions are well known in most cases. 

However, in some scenarios, 
the decoder receives each codeword symbol \emph{sequentially}, i.e., each codeword symbol is received in some gradual manner, rather than instantaneously.
When these scenarios involve codes over~$\bF_{q^\alpha}$, codeword symbols are viewed as vectors over~$\bF_q$, and the decoder receives these vectors one~$\bF_q$ element after another. In this paper we study bounds and code constructions for this scenario. 
That is, codes that enable the decoder to complete the decoding process once sufficiently many~$\bF_q$ symbols are obtained regardless of their source, and in particular, even if $\bF_{q^\alpha}$-symbols have not been obtained in full. Practical applications which present this behavior, for which our techniques are useful, are discussed in the sequel.

In the next section we lay the mathematical framework by which we study the problem, discuss potential applications, and summarize our contributions. Several constructions of codes capable of correcting hierarchical erasures are given in Section~\ref{section:constructions} while upper and a lower bound is discussed in Section~\ref{section:bounds}.

\section{Preliminaries}
\subsection{Framework and Problem Definition}
Let~$\boldc=(c_i)_{i=1}^n\in\bF_{q^\alpha}^n$ be a codeword in a linear code. By fixing a basis\footnote{Typically, bases are considered as sets, not as vectors. In this paper however, we consider bases of~$\bF_{q^\alpha}$ over~$\bF_q$ as (row) vectors of length~$\alpha$ over~$\bF_{q^\alpha}$, the entries of whom span~$\bF_{q^\alpha}$ over~$\bF_q$.}~$\boldomega=(\omega_1,\ldots,\omega_\alpha)$ of~$\bF_{q^\alpha}$ over~$\bF_q$, consider each~$c_i$ as a vector in~$\bF_q^\alpha$, and denote (by abuse of notation)~$c_i=(c_{i,1},\ldots,c_{i,\alpha})$ where $c_i=\sum_{j=1}^\alpha c_{i,j}\omega_j$.

For an integer~$m$, an~$m$-hierarchical erasure in~$\boldc$ amounts to erasing at most~$m$ \emph{left-justified} entries of all~$c_i$'s. That is, for every $m$-hierarchical erasure in~$\boldc$, there exists a tuple~$(t_1,\ldots,t_n)$ of nonnegative integers whose sum is at most~$m$ such that~$c_{1,1},\ldots,c_{1,t_1},c_{2,1},\ldots,c_{2,t_2},\ldots,$ $c_{n,1},\ldots,c_{n,t_n}$ are replaced by~$*$. For example, for~$\alpha=3$, $n=4$, and~$m=5$, all of the following are examples of~$m$-hierarchical erasures in a codeword~$\boldc\in\bF_{q^3}^4$:
\begin{align}\label{equation:HierarchicalExample}
    &\left( (*,c_{1,2},c_{1,3}),
    (*,*,c_{2,3}),
    (*,c_{3,2},c_{3,3}),
    (*,c_{4,2},c_{4,3})\right)\nonumber\\
    &\left( (c_{1,1},c_{1,2},c_{1,3}),
    (*,*,*),
    (*,c_{3,2},c_{3,3}),
    (*,c_{4,2},c_{4,3})\right)\nonumber\\
    &\left( (*,*,c_{1,3}),
    (c_{2,1},c_{2,2},c_{2,3}),
    (*,*,c_{3,3}),
    (*,c_{4,2},c_{4,3})\right).
\end{align}

In contract, the following is \emph{not} a hierarchical erasure, since the erasures are not left-justified:
\begin{align*}
    \left( (c_{1,1},*,c_{1,3}),
    (*,c_{2,2},c_{2,3}),
    (*,*,c_{3,3}),
    (*,c_{4,2},c_{4,3})\right).
\end{align*}

 Given a basis~$\boldomega$ of~$\bF_{q^\alpha}$ over~$\bF_q$, a linear code~$\cC$ is called \emph{an~$m$-correcting code over~$\boldomega$} if it is possible to correct any~$m$-hierarchical erasure, where codeword symbols are represented in the basis~$\boldomega$. The goal of this paper is, given the parameters~$n$, $m$, and~$\alpha$, to find a basis~$\boldomega$ and construct a linear~$m$-correcting code over~$\boldomega$, with maximum dimension~$k$ and minimum base-field size~$q$.

For positive integers $\alpha,n,$ and~$m$ let
\begin{align*}
    \cN_{\alpha,m}^n\triangleq\set*{(t_1,t_2,\ldots,t_n) ; \text{$0\le t_i\le \alpha$ for all $i$  and $\sum_{i=1}^nt_i\le m$}}.
\end{align*}
In the special case where~$\alpha=m$ we use the shorthand notation~$\cN_\alpha^n$. An element~$\boldt\in\cN_{\alpha,m}^n$ is called \emph{an erasure pattern}, and it uniquely determines the locations of the~$*$ symbols in a hierarchical erasure. For instance, the erasure patterns which appear in~\eqref{equation:HierarchicalExample} are~$(1,2,1,1)$, $(0,3,1,1)$, and $(2,0,2,1)$, respectively. For a set~$\cT\subseteq \cN_{\alpha,m}^n$, we say that~$\cC\subseteq \bF_{q^\alpha}^n$ is~$\cT$-correcting over~$\boldomega$ if all erasure patterns in~$\cT$ can be corrected. An~$\cN_{\alpha,m}^n$-correcting code is called an~$m$-correcting code.

We make repeated use of the following notations. For an integer~$\ell$ let~$[\ell]\triangleq \set{1,2,\ldots,\ell}$. For~$\boldc\in \bF_{q^\alpha}^\ell$ and a basis~$\boldomega$ of~$\bF_{q^\alpha}$ over~$\bF_q$ let 
\begin{align*}
    w_{\boldomega}(\boldc)&\triangleq \sum_{i\in[\ell]}\max
    \set*{ j\in[\alpha] ; c_{i,j}\ne 0 },
\end{align*}
where the~$c_{i,j}$'s are the coefficients of the entries~$\boldc$ in the representation over~$\boldomega$, as explained above, and the subscript~$\boldomega$ is omitted if clear from the context. 

Finally, we note that to the best of our knowledge, this paper is the first to study linear hierarchical erasure correcting codes. Yet, similar problems have been studied in the past. The closest one is~\cite{RosTsf97}, in which exactly the same erasure patterns have been studied, bounds formulated, and constructions provided. However, the codes there are linear \emph{after} having each element from $\bF_{q^\alpha}$ expanded to its coordinate vector of length $\alpha$ over $\bF_q$ in some basis $\boldomega$. But when considered as a code over $\bF_{q^\alpha}$, the code is closed under addition and multiplication only by scalars from $\bF_q$, and not necessarily under multiplication by scalars from $\bF_{q^\alpha}$, namely, it is not necessarily linear. Such codes are sometimes referred to as \emph{vector-linear} codes. This work was later generalized in~\cite{BarPar15}, but still under the vector-linear coding framework. In another recent work~\cite{TamYeBar17}, the decoder does not access the entire $\bF_{q^\alpha}$ code symbol, but unlike our paper, it is allowed to freely choose the function to extract from the symbol.


\subsection{Potential Applications}\label{subsection:applications}
Linear codes have widespread applications in coding for distributed storage systems~\cite{plank2005t1}. Normally, a database~$\boldx\in\bF_{q^\alpha}^k$ is mapped to a codeword~$\boldc\in\bF_{q^\alpha}^n$, and each codeword symbol is stored on a different storage server. Then, in cases where some servers might be unavailable due to hardware failures, the reconstruction of the entire database~$\boldx$ by communicating with the storage servers corresponds to (ordinary) erasure correction. 

However, it has been demonstrated recently that modern distributed systems are prone to the \emph{stragglers} phenomenon~\cite{dean2013tail}, which are servers that respond much slower than the average. Moreover, communicating a large amount of data from a server does not occur instantaneously, but rather as an ordered sequence of bits or packets. Therefore, it is evident that our problem is directly applicable to storage systems that employ linear codes, and suffer from the straggler phenomenon. For applications of this sort, one might be more interested in the regime~$\alpha\gg n$, since the number of storage servers in the systems is likely to be much smaller than the content of each individual server. 

Additional applications can be found in flash storage devices that employ \emph{low-density parity-check} (LDPC) codes. A flash memory cell can store~$2^\alpha$ distinct charge levels, each representing a stored binary vector of length $\alpha$. Reading the cell can be done by applying a series of~$2^\alpha-1$ threshold tests, ordered in a way that recovers the $\alpha$ bits one after another\footnote{While a single cell may be tested using only $\alpha$ tests using a binary-search algorithm, in a typical flash memory a threshold test is administered to a large array of cells at once. Thus, typically, some cells in the array would test below the threshold and some above. To find out the charge levels in all the cells we would typically need to test all $2^{\alpha}-1$ thresholds. Nonetheless, the thresholds may be ordered to test at $1/2$-range, $1/4$-range, $3/4$-range, and so on, making the first test obtain the most-significant bit of each cell, the following two tests to obtain the second-most-significant bit, and so on.}. In the event that this series of threshold tests discontinues abruptly due to hardware failures, the missing bits from the readout value correspond to a hierarchical erasure. A common and effective approach to decoding LDPC codes consists of \emph{variable nodes}, representing the codeword symbols, and \emph{check nodes}, which represent a linear combination of variable nodes. Then, decoding is performed in an iterative manner, where variable nodes communicate with check nodes and vice versa~\cite{Gal63}.

Each check node represents an equation~$\sum_{i=1}^n h_ix_i=0$, where each~$x_i\in\bF_{q^\alpha}$ is a variable node representing a value contained in a flash memory cell, and the~$h_i$'s are pre-determined coefficients in~$\bF_{q^\alpha}$. It is readily verified that if the right kernel of the row vector~$\boldh=(h_i)_{i=1}^n$ is an~$m$-correcting code, one can resolve any~$m$-hierarchical erasure in the code symbols~$(x_1,\ldots,x_n)$. For applications of this sort, one might be more interested in the regime~$n\gg\alpha$, since the typical number of bits stored per cell is much smaller than a useful codeword length~$n$. Decoding of LDPC codes with $m$-correcting check nodes was studied in~\cite{CohCas16,CohRavCas19}, which served as the main inspiration for the current paper.

\subsection{Universally Decodable Matrices}
The problems in this paper are intimately connected to \emph{Universally Decodable Matrices} (UDMs)~\cite{GanVon07,VonGan06}, which are a useful tool in error correction of \emph{slow-fading channels}~\cite{TavVis06}.


\begin{definition}[{\cite[Def.~1]{GanVon07}}]
	For~$m\ge \alpha$, matrices $A_1,\ldots,A_n\in\bF_q^{\alpha\times m}$ are called Universally Decodable Matrices (UDMs) if for every~$\boldt=(t_1,\ldots,t_n)\in\cN_{\alpha,m}^n$ the following condition is satisfied: the matrix composed of the first~$t_1$ rows of~$A_1$, the first~$t_2$ rows of~$A_2$, ..., the first~$t_n$ rows of~$A_n$, has full rank.
\end{definition}

In the following theorem let~$I_{\alpha\times m}$ be the first~$\alpha$ rows of an~$m\times m$ identity matrix. Similarly, let~$J_{\alpha\times m}$ be the first~$\alpha$ rows in the anti-identity matrix, i.e., the matrix which contains~$1$'s in its anti-diagonal, and zero elsewhere.

\begin{theorem}[{\cite[Prop.~14]{VonGan06}}]\label{theorem:Vontobel}
	Let~$n,m$, and $\alpha$ be positive integers, let~$q$ be a prime power such that~$q\ge n-1$, and let~$\gamma$ be a primitive element in~$\bF_q$. Then, the following are~$\alpha\times m$ UDMs over~$\bF_q$
	\begin{align*}
	&A_0\triangleq I_{\alpha,m},~A_1\triangleq J_{\alpha,m},A_2,\ldots,A_{n-1}\mbox{ where}\\
	&(A_{i+1})_{a,b}=\binom{b}{a}\gamma^{(i-1)(b-a)}\mbox{ for }(i,a,b)\in [n-2]\times [\alpha]\times [m].
	\end{align*}
\end{theorem}

UDMs will be used in Subsection~\ref{subsection:fromTraces} to define the parity check matrix of the constructed codes. Further, in Appendix~\ref{appendix:mutualEigenvector} it is shown that the important special case~$\alpha=m$ is tightly connected to the existence of UDMs with a certain mutual eigenvector. 

\subsection{Main Lemma}
Most of the results in this paper are based on the following lemma. It is stated generally for~$\cT$-correcting codes for any~$\cT\subseteq \cN_{\alpha,m}^n$, and specifies to~$m$-correcting code by choosing~$\cT=\cN_{\alpha,m}^n$. For an erasure pattern~$\boldt\in\cN_{\alpha,m}^n$ and a basis~$\boldomega$, denote
\begin{align}\label{equation:Xt}
\cX_\boldt=\cX_\boldt(\boldomega )\triangleq \;&\Span{\{(\omega_i,0,\ldots,0)\}_{i\in[t_1]}}  \oplus
\nonumber \\
& \Span{\{(0,\omega_i,0,\ldots,0)\}_{i\in[t_2]}} \oplus \nonumber\\
&\cdots\nonumber\\
&\Span{\{(0,\ldots,0,\omega_i)\}_{i\in[t_n]}} ,
\end{align}
where each vector in~\eqref{equation:Xt} is of length~$n$, $\Span{\cdot}$ denotes span over~$\bF_q$, and~$\oplus$ is the sum of subspaces that intersect trivially. For example, for~$n=3$, $m=4$, and~$\boldt=(2,1,1)\in\cN_{2,4}^3$ we have $\cX_\boldt= \Span{(\omega_1,0,0),(\omega_2,0,0),(0,\omega_1,0),(0,0,\omega_1)} $. Note that the elements of~$\cX_\boldt$ are precisely the ones that are indistinguishable from the zero vector under the erasure pattern~$\boldt$.

\begin{lemma}\label{lemma:main}
    For any~$\cT\subseteq \cN_{\alpha,m}^n$, a linear code~$\cC\subseteq \bF_{q^\alpha}^n$ is~$\cT$-correcting over~$\boldomega$ if and only if~$\cC\cap \cX_\boldt=\{0\}$ for every~$\boldt\in\cT$.
\end{lemma}
\begin{proof}
    To prove one direction, assume that~$\cC$ is~$\cT$-correcting. If~$\cC$ contains a nonzero codeword which belongs to~$\cX_\boldt$ for some~$\boldt\in\cT$, then this codeword is indistinguishable from the zero word under the erasure pattern~$\boldt$, which implies that~$\boldt$ is not correctable. 
    
    Conversely, assume that~$\cC\cap \cX_\boldt=\{0\}$ for every~$\boldt\in\cT$. If~$\cC$ is not~$\cT$-correcting, it follows that there exist two distinct words
    \begin{align*}
        \boldc^{(1)} &= \left( (c^{(1)}_{1,1},\ldots,c^{(1)}_{1,\alpha}),\ldots,(c^{(1)}_{n,1},\ldots,c^{(1)}_{n,\alpha}) \right)\\
        \boldc^{(2)} &= \left( (c^{(2)}_{1,1},\ldots,c^{(2)}_{1,\alpha}),\ldots,(c^{(2)}_{n,1},\ldots,c^{(2)}_{n,\alpha}) \right)
    \end{align*}
    that are indistinguishable after some erasure pattern~$\boldt=(t_i)_{i=1}^n \in\cT$. This indistinguishability implies that~$c_{i,j}^{(1)}=c_{i,j}^{(2)}$ for every~$(i,j)\in [n]\times ([\alpha]\setminus [t_i])$; and since the code is linear, it follows that~$\boldd\triangleq \boldc^{(1)}-\boldc^{(2)}$ belongs to~$\cC$ as well. However, it is readily verified that~$\boldd$ is a nonzero codeword in~$\cC\cap \cX_\boldt$, a contradiction.
\end{proof}

\subsection{Our Contribution}
We begin in Subsection~\ref{subsection:alphaCorrectingofLengthTwo} with a construction for the parameters~$(n,k,m)=(2,1,\alpha)$. The well-known trace operator is used in Subsection~\ref{subsection:fromTraces} to construct $m$ correcting codes that are better suited for the regime~$n\gg \alpha$. 

Since extending these two constructions to other parameters proved difficult, in Subsection~\ref{subsection:balanced} we resort to restricted types of erasure patterns called \textit{balanced} and the important case~$k=n-1$, which generalizes the prevalent parity code. In Subsection~\ref{subsection:MoshesSubsection} we discuss \textit{power} erasure patterns, that generalize the balanced ones, and provide a code construction for~$k=n-1$ at the price of a larger base field than for balanced patterns. We conclude the constructive part of the paper in Subsection~\ref{subsection:Gabidulin}, by showing that Gabidulin codes can correct yet another restricted type of erasure patterns. The parameters for all the constructions in this paper are given in Table~\ref{tab:results}. Finally, several simple upper bounds and an existential lower bound are given in Section~\ref{section:bounds}.


\begin{table}[t]
\begin{tabular}{|c|c|c|c|c|}
\hline
\cellcolor[HTML]{C0C0C0}Subsection & \cellcolor[HTML]{C0C0C0}Field & \cellcolor[HTML]{C0C0C0}Parameters & \cellcolor[HTML]{C0C0C0}Patterns & \cellcolor[HTML]{C0C0C0}Tool \\ \hline\hline
\ref{subsection:alphaCorrectingofLengthTwo} & Any  & \makecell{$n=2$\\ $k=1$\\ $m=\alpha$ even} & $\cN_{\alpha}^2$ & Irreducible polynomial \\ \hline
\ref{subsection:fromTraces} & $q\ge n-1$ & $k\ge n-m$ & $\cN_{\alpha,m}^n$ & Trace, dual bases \\ \hline
\ref{subsection:balanced} & \makecell{$q\ge n-1$\\ $\alpha=2^\beta$} & \makecell{$k=n-1$\\ $m=\alpha$} & $\cN^n_{\alpha\vert\bal}$ & Subfield independence \\ \hline
\ref{subsection:MoshesSubsection} &\makecell{ $q\ge \frac{\alpha}{2} n+1$\\ $\alpha=2^\beta$\\ $\frac{\alpha}{2}|q-1$ } &  \makecell{$k=n-1$\\ $m=\alpha$} & $\cN^n_{\alpha\vert\pow}$ & Determinant \\ \hline
\ref{subsection:Gabidulin} & Any &\makecell{ $k=n-r$\\ $\alpha\ge n\ge r$} &$\cN_{r,nr}^n$  & Gabidulin codes \\ \hline
\end{tabular}
\caption{Summary of constructions.}
\label{tab:results}
\end{table}

\section{Constructions}\label{section:constructions}
\subsection{\texorpdfstring{$\alpha$}{alpha}-correcting codes of length two} \label{subsection:alphaCorrectingofLengthTwo}
\begin{theorem}\label{theorem:n=2}
For any prime power~$q$ and any even~$\alpha\in\bN$, the code
\begin{align*}
	\cC\triangleq \set*{ \boldc\in \bF_{q^\alpha}^2 ; (1,b)\cdot \boldc^\intercal=0 }
\end{align*}
is~$\alpha$-correcting, where~$b$ is a root of an irreducible quadratic polynomial over~$\bF_q$.
\end{theorem}
To prove this theorem, the following lemmas are given. In what follows, for an element~$b\in\bF_{q^\alpha}$ and an even~$\alpha$, a basis~$\boldomega =( \omega_1,\ldots,\omega_{\alpha} )$ of~$\bF_{q^\alpha}$ over~$\bF_q$ is called~\emph{$b$-symmetric} if~$\omega_{\alpha-i+1}=b\omega_i$ for all~$i\in[\alpha/2]$; namely, if
\begin{align*}
    \boldomega = (\omega_1,\omega_2,\ldots,\omega_{\alpha/2},b\omega_{\alpha/2},\ldots,b\omega_2,b\omega_1).
\end{align*}

\begin{lemma}\label{lemma:Yuval2}
	For any even~$\alpha\in\bN$ and any prime power~$q$, there exists a~$b$-symmetric basis of~$\bF_{q^\alpha}$ over~$\bF_q$, where~$b\in\bF_{q^\alpha}$ is a root of an irreducible quadratic polynomial~$P(x)$ over~$\bF_q$.
\end{lemma}

\begin{proof}
	Denote~$\alpha=2^t\ell$, where~$\ell$ is odd and~$t\ge 1$. We prove this claim by induction on~$t$. For~$t=1$ let~$\omega_1,\ldots,\omega_\ell$ be a basis of~$\bF_{q^\ell}$ over~$\bF_q$. Notice that~$P(x)$ remains irreducible when seen as a polynomial over~$\bF_{q^\ell}$; otherwise, we have that~$P(x)$ is a minimal polynomial of some element in~$\bF_{q^\ell}$, whose degree does not divide~$\ell$, a contradiction. Hence, we have that~$b\notin\bF_{q^\ell}$, and thus~$(\omega_1,\ldots,\omega_\ell,b\omega_\ell,\ldots,b\omega_1 )$	is a~$b$-symmetric basis of~$\bF_{q^\alpha}$ over~$\bF_q$.
	
	For~$t>1$, by the induction hypothesis there exists a~$b$-symmetric basis $(\omega_1,\ldots,\omega_{\alpha/2})$ of~$\bF_{q^{\alpha/2}}$ over~$\bF_q$. By choosing any~$\gamma\in\bF_{q^\alpha}\setminus\bF_{q^{\alpha/2}}$, it is readily verified that
	\begin{align*}
	\boldomega &\triangleq(\gamma\omega_1,\omega_1,\ldots,\gamma\omega_{\alpha/4},\omega_{\alpha/4},
	\omega_{{\alpha/4}+1},\gamma\omega_{{\alpha/4}+1},\ldots,\omega_{\alpha/2},\gamma\omega_{\alpha/2})\\
	&=(\gamma\omega_1,\omega_1,\ldots,\gamma\omega_{\alpha/4},\omega_{\alpha/4},
	b\omega_{\alpha/4},b\gamma\omega_{\alpha/4},\ldots,b\omega_1,b\gamma\omega_1)
	\end{align*}
	is a~$b$-symmetric basis of~$\bF_{q^\alpha}$ over~$\bF_q$, where the last equality follows from the induction hypothesis.
\end{proof}

\begin{lemma}\label{lemma:Yuval1}
	If~$\boldomega =(\omega_i )_{i\in[\alpha]}$ is a~$b$-symmetric basis, with $b\in\bF_{q^\alpha}$ being a root of an irreducible quadratic polynomial~$P(x)=x^2+a_1x+a_0$ over~$\bF_q$, then
	\begin{align*}
		\Span{b\omega_1,b\omega_2,\ldots,b\omega_t} =\Span{\omega_\alpha,\omega_{\alpha-1}\ldots,\omega_{\alpha-t+1}} 
	\end{align*}
	for every $t\in [\alpha]$.
\end{lemma}
\begin{proof}
	If~$t\le \alpha/2$, then the claim follows from the definition of a~$b$-symmetric basis. If~$t \ge \alpha/2+1$, we have that
	\begin{align*}
		\Span{b\omega_1,\ldots,b\omega_t}
		&=\Span{ \{ b\omega_i \}_{i=1}^{\alpha/2} } +\Span{ \{ b\omega_i \}_{i=\alpha/2+1}^t }\\
		&=\Span{ \{ \omega_i \}_{i=\alpha/2+1}^{\alpha} } +\Span{ \{ b^2\omega_{\alpha-i+1} \}_{i=\alpha/2+1}^t } \\
		&=\Span{ \{ \omega_i \}_{i=\alpha/2+1}^{\alpha} } +\Span{ \{ (-a_1b-a_0)\omega_{\alpha-i+1} \}_{i=\alpha/2+1}^t } \\
		&=\Span{ \{ \omega_i \}_{i=\alpha/2+1}^{\alpha} } +\Span{ \{ -a_1\omega_i-a_0\omega_{\alpha-i+1} \}_{i=\alpha/2+1}^{t} } \\		
		&=\Span{ \{ \omega_i \}_{i=\alpha/2+1}^{\alpha} } +\Span{ \{ \omega_i \}_{i=\alpha-t+1}^{\alpha/2} } \\
		&=\Span{\omega_\alpha,\omega_{\alpha-1},\ldots,\omega_{\alpha-t+1}}.\qedhere
	\end{align*}
\end{proof}
Lemma~\ref{lemma:Yuval2} and Lemma~\ref{lemma:Yuval1} imply Theorem~\ref{theorem:n=2} as follows.
\begin{proof}(of Theorem~\ref{theorem:n=2})
	Let~$\boldomega $ be a~$b$-symmetric basis of~$\bF_{q^\alpha}$ over~$\bF_q$, as guaranteed by Lemma~\ref{lemma:Yuval2}. According to Lemma~\ref{lemma:main}, it suffices to prove that~$\cC\cap \cX_\boldt=\{0\}$ for every~$\boldt\in\cN_\alpha^2$. Assume to the contrary that there exists~$\boldt\in\cN_{\alpha}^2$ and a nonzero codeword~$\boldc=(c_1,c_2)\in\cC$ such that~$\boldc\in \cX_\boldt(\boldomega )$. This readily implies that
	\begin{align}
		c_1&\in\Span{\omega_1,\ldots,\omega_{t_1}} \label{equation:Thm1row1},\\
		c_2&\in\Span{\omega_1,\ldots,\omega_{t_2}} \mbox{, and}\label{equation:Thm1row2}\\
		c_1&+bc_2=0\label{equation:Thm1row3}.
	\end{align}
	Furthermore, Lemma~\ref{lemma:Yuval1} and Eq.~\eqref{equation:Thm1row2} imply that~$bc_2$ is in~$\Span{\omega_{\alpha},\omega_{\alpha-1},\ldots,\omega_{\alpha-t_2+1}}$. Since~$t_1+t_2<\alpha+1$, it follows that~$t_1<\alpha-t_2+1$, and hence~\eqref{equation:Thm1row1} implies that~\eqref{equation:Thm1row3} is a sum of elements from trivially intersecting subspaces that results in zero, and hence~$c_1$ and~$bc_2$ must both be zero. Since~$b$ is nonzero, this implies that~$(c_1,c_2)=(0,0)$, a contradiction.
\end{proof}

\begin{remark}
    An alternative proof for this construction can be obtained by viewing it as a pair of UDMs with the added property that they share an eigenvector whose entries span~$\bF_{q^\alpha}$ over~$\bF_q$. More details on this view (for general $n\geq 2$) are given in Appendix~\ref{appendix:mutualEigenvector}.
\end{remark}


\subsection{\texorpdfstring{$m$}{m}-correcting codes from traces}\label{subsection:fromTraces}

In this section we make use of the \emph{trace operator}~$\Tr$~\cite[Def.~2.22]{LidNie97} and \emph{dual bases}~\cite[Def.~2.30]{LidNie97}. The \emph{trace} of an element~$c\in\bF_{q^\alpha}$ (with respect to $\bF_q$) is defined as
\[\Tr(c)\triangleq c+c^q+c^{q^2}+\ldots+c^{q^{\alpha-1}}.\]
The trace function is linear over~$\bF_q$, i.e.,~$\Tr(\gamma a+\delta b)=\gamma\Tr(a)+\delta\Tr(b)$ for every~$\gamma,\delta\in\bF_q$ and~$a,b\in\bF_{q^\alpha}$. Two bases~$\boldomega=(\omega_i)_{i=1}^\alpha$ and~$\boldmu=(\mu_i)_{i=1}^\alpha$ are called dual if
\begin{align*}
    \Tr(\omega_i\cdot\mu_j)=
    \begin{cases}
        0 & \text{if $i\ne j$,}\\
        1 & \text{if $i= j$,}
    \end{cases}
\end{align*}
and for every basis there exists a unique dual basis~\cite[Def.~2.30]{LidNie97}.

\begin{theorem}\label{theorem:anyBasis}
	For positive integers $m\geq\alpha$, let~$\{A_i\}_{i\in [n]}$ be~$\alpha\times m$ UDMs over~$\bF_q$, and let~$\boldmu$ be a basis of~$\bF_{q^\alpha}$ over~$\bF_q$. Then, the code
	\begin{align*}
		\cC \triangleq \set*{ (c_1,\ldots,c_n)\in\bF_{q^\alpha}^n ; \left( A_1^\intercal \boldmu^\intercal \vert \cdots \vert A_n^\intercal \boldmu^\intercal \right)\cdot \left( c_1,\ldots,c_n \right)^\intercal = 0 }
	\end{align*}
	is~$m$-correcting over the dual~$\boldomega$ of~$\boldmu$, and~$\dim\cC\ge n-m$.
\end{theorem}

\begin{proof}
	Assume to the contrary that there exists~$\boldt\in\cN_{\alpha,m}^n$ and a nonzero codeword~$\boldc\in\cC$ such that~$\boldc\in \cX_\boldt(\boldomega)$. Therefore, any codeword symbol~$c_i$ can be written as~$c_i=\sum_{j\in[t_i]}c_{i,j}\omega_j$ for some coefficients~$c_{i,j}\in\bF_q$, and hence
	\begin{align*}
		\boldmu^\intercal c_i = 
		\begin{pmatrix}
			\sum_{j\in[t_i]}c_{i,j}\omega_j\mu_1\\
			\sum_{j\in[t_i]}c_{i,j}\omega_j\mu_2\\
			\vdots\\
			\sum_{j\in[t_i]}c_{i,j}\omega_j\mu_\alpha\\
		\end{pmatrix}.
	\end{align*}
	Thus, for every~$\ell\in[m]$, the~$\ell$'th entry of the equation~$\sum_{i\in[n]}A_i^\intercal\boldmu^\intercal c_i=0$ equals
	\begin{align*}
		\sum_{i\in[n]}\sum_{r\in[\alpha]}A_i^{(r,\ell)}\sum_{j\in[t_i]}c_{i,j}\omega_j\cdot\mu_r=0,
	\end{align*}
	where~$A_i^{(r,\ell)}$ is the~$(r,\ell)$'th entry of~$A_i$. Applying the trace function on both sides, and exploiting the linearity of the trace and the fact that~$\boldomega$ and~$\boldmu$ are dual, yields
	\begin{align*}
		\sum_{i\in[n]}\sum_{r\in[t_i]}A_i^{(r,\ell)}c_{i,r}=0\mbox{ for every }\ell\in[m].
	\end{align*}
	In turn, this implies that the vector~$(c_{1,1},\ldots,c_{1,t_1},\cdots,c_{n,1},\ldots,c_{n,t_n})$ is in the left kernel of
	\begin{align*}
		\begin{pmatrix}
		A_1^{(1:t_1)}\\
		A_2^{(1:t_2)}\\
		\vdots\\
		A_n^{(1:t_n)}\\
		\end{pmatrix},
	\end{align*}
	where~$A_i^{(1:t_i)}$ is a matrix which contains the top~$t_i$ rows of~$A_i$, which contradicts the definition of UDMs. The bound~$\dim \cC\ge n-m$ follows since~$\cC$ is the right kernel of an~$m\times n$ matrix.
\end{proof}

In light of the bound~$\dim\cC\ge n-m$ that is given above, one might prefer to employ this construction in the regime~$n\gg \alpha$. However, for the case of even~$m=\alpha=n$, one can guarantee~$\dim\cC>0$ by using techniques from Subsection~\ref{subsection:alphaCorrectingofLengthTwo}. The proof is given in Appendix~\ref{appendix:n=2Extension}.


\begin{corollary}\label{corollary:n=2Extension}
    For even~$m=\alpha=n\in\bN$, let~$\{A_i\}_{i=1}^n$ be $\alpha\times\alpha$ UDMs such that~$A_1$ is the identity matrix, and
    \begin{align*}
        A_2=
        \begin{pmatrix}
            ~&~&~&~&~&1\\
            ~&~&~&~&\iddots&~\\
            ~&~&~&1&~&~\\
            ~&~&-a_0&-a_1&~&~\\
            ~&\iddots&~&~&\ddots&~\\
            -a_0&~&~&~&~&-a_1\\
        \end{pmatrix},
    \end{align*}
    where~$x^2+a_1x+a_0$ is an irreducible quadratic polynomial over~$\bF_q$ with a root ~$b\in\bF_{q^\alpha}$. In addition, let~$\boldmu$ be a~$b$-symmetric basis (see Lemma~\ref{lemma:Yuval2}), and let~$\boldomega$ be its dual. Then, the code
    \begin{align*}
		\cC \triangleq \set*{ (c_1,\ldots,c_n)\in\bF_{q^\alpha}^n ; \left( A_1^\intercal \boldmu^\intercal \vert \cdots \vert A_n^\intercal \boldmu^\intercal \right)\cdot \left( c_1,\ldots,c_n \right)^\intercal = 0 }
	\end{align*}
	is an~$\alpha$-correcting code over~$\boldomega$ with~$\dim \cC\ge 1$.
\end{corollary}

\subsection{Correcting balanced erasure patterns}\label{subsection:balanced}
The case~$k=n-1$ and~$m=\alpha$ is of particular importance, since it generalizes the widely used parity code (for storage applications), and corresponds to hierarchical erasure correction in check nodes of LDPC codes (see Subsection~\ref{subsection:applications}). This case is not handled well by previous subsections; in Subsection~\ref{subsection:alphaCorrectingofLengthTwo} it necessitates~$n=2$ (i.e., a short code), and in Subsection~\ref{subsection:fromTraces} one must have~$m=1$ (i.e., low erasure correction) to get $k=n-1$. Hence, in this subsection we focus on this case, and show a code construction which protects against erasure patterns that we call \emph{balanced}. This case is also addressed in Subsection~\ref{subsection:MoshesSubsection} which follows, where a stronger erasure correction is guaranteed at the price of a larger base field, by using similar techniques.

Assume that~$\alpha=2^\beta$ for some integer~$\beta$. An erasure pattern~$\boldt\in\cN_\alpha^n$ is called balanced if there exists an integer~$0\le i\le \min\{ \beta,\log n \}$ 
(where the logarithm is to base~$2$) and a set~$J\subseteq [n]$ with~$|J|\leq 2^i$, such that for all~$j\in[n]$,
\begin{align*}
	\begin{cases}
			t_j \leq \frac{\alpha}{2^i} &\mbox{ if } j\in J\mbox{; and}\\
			t_j = 0 &\mbox{otherwise.}
	\end{cases}
\end{align*}
For example, for~$n=4$ the erasure patterns 
\begin{align*}
    &(\alpha/2,0,\alpha/2,0),\mbox{ and}\\
    &(\alpha/4,\alpha/4,\alpha/4,\alpha/4)
\end{align*}
are balanced, whereas~$(\alpha/2,\alpha/4,\alpha/4,0)$ is not. The set of all balanced erasure patterns is denoted by~$\cN_{\alpha\vert \bal}^n$.

We consider bases~$\boldomega =(\omega_1,\ldots,\omega_\alpha)$ of~$\bF_{q^\alpha}$ over~$\bF_q$ that we call \emph{recursive}, i.e., bases such that~$\Span{\omega_1,\ldots,\omega_{\alpha/2^i}} =\bF_{q^{\alpha/2^i}}$ for all~$0\le i\le\beta$. For a vector~$\boldh=(h_1,\ldots,h_n)\in\bF_{q^\alpha}^n$ we define a code
\begin{equation}
    \label{equation:cdef}
    \cC=\cC(\boldh)\triangleq \ker(\boldh)\triangleq\set*{\boldc\in\bF_{q^\alpha} ; \boldh\boldc^\intercal=0}.
\end{equation}
The ability of the code $\cC$ to protect against balanced erasure patterns reduces to linear independence of some subsets of the~$h_i$'s over certain subfields of~$\bF_{q^\alpha}$, as we now show.

%

\begin{lemma}\label{lemma:linearIndependence}
	The code~$\cC$~\eqref{equation:cdef} is~$\cN_{\alpha\vert\bal}^n$-correcting over a recursive basis~$\boldomega $ if and only if for every~$1\le i \le \min\{\beta, \log n \}$, every~$2^i$-subset of~$\{h_j\}_{j\in[n]}$ is a linearly independent set over~$\bF_{q^{\alpha/2^i}}$.
\end{lemma}
\begin{proof}
	Assume that every~$2^i$-subset of~$\{h_j \}_{j=1}^n$ is linearly independent over $\bF_{q^{\alpha/2^i}}$ for every~$0\le i\le \min\{ \beta,\log n \}$. According to Lemma~\ref{lemma:main}, if~$\cC$ is not~$\cN_{\alpha\vert\bal}^n$-correcting, then there exists a nonzero~$\boldc=(c_1,c_2,\ldots,c_n)$ in~$\cC$ and an erasure pattern~$\boldt\in \cN_{\alpha\vert\bal}^n$ such that~$\boldc\in\cC\cap \cX_\boldt$. By the definition of~$\cN_{\alpha\vert\bal}^n$, it follows that there exists an integer~$i$ and a set~$J\subseteq[n]$ of size at most~$2^i$ such that~$t_j\le\alpha/2^i$ if~$j\in J$, and $t_j=0$ otherwise. Hence, we have that
	\begin{align*}
		c_j\in\Span{\omega_1,\ldots,\omega_{\alpha/2^i}}=\bF_{q^{\alpha/2^i}}\mbox{ for all }j\in J,
	\end{align*}
	which implies that~$\sum_{j\in J}h_jc_j=0$. However, this sum is a linear combination of a~$2^i$-subset of~$\{ h_j \}_{j\in[n]}$ over~$\bF_{q^{\alpha/2^i}}$, a contradiction. The proof of the inverse direction is similar.
\end{proof}

In what follows we construct an~$[n,n-1]_{q^\alpha}$~$\cN_{\alpha\vert\bal}^n$-correcting code, for any~$n$ and any~$\alpha$ over a base field~$\bF_q$ with~$q\ge n-1$. To this end, recall that $\alpha=2^\beta$, and let~$\{b_i \}_{i\in[\beta]}\subseteq \bF_{q^\alpha}$ such that
\begin{equation}
\label{equation:defbi}
\bF_{q^{\alpha/2^{i-1}}}=\bF_{q^{\alpha/2^i}}(b_i),
\end{equation}
for all~$i\in[\beta]$, i.e., we consider each subfield~$\bF_{q^{\alpha/2^{i-1}}}$ as a vector space of dimension two over~$\bF_{q^{\alpha/2^i}}$ by fixing the basis~$\{ 1,b_i \}$. 

For~$0\le i\le \beta$ and a~$2^i\times n$ matrix~$M$ over~$\bF_{q^{\alpha/2^i}}$, let
\begin{align*}
\cH_i(M)\triangleq\UH(M)+b_i\LH(M),
\end{align*}
where~$\UH$ and~$\LH$ denote the upper half and lower half of~$M$, respectively. Further, for an integer~${1\le i\le \beta}$ and an~$\alpha\times n$ matrix~$M$ over~$\bF_q$ let
\begin{align*}
\cH^{(i)}(M) &\triangleq \cH_{\beta-i+1}(\cdots(\cH_{\beta-1}(\cH_{\beta}(M)))),\\
\cH^{(0)}(M)&\triangleq M.
\end{align*}

Throughout the remainder of this section we use a recursive basis induced by the~$\{ b_i \}_{i\in[\beta]}$ from~\eqref{equation:defbi}. Namely, the basis is
\begin{equation}
\label{equation:defrecbase}
\boldomega \triangleq W_{\beta},\text{ where}\quad
W_0 \triangleq (1),\text{ and}\quad
W_{i+1}\triangleq W_{i}\vert(b_{\beta-i}\cdot W_i),
\end{equation}
and~$\vert$ denotes concatenation. Alternatively,
\[ \boldomega\triangleq
(1,b_1)
\otimes
(1,b_2)
\otimes
\cdots
\otimes
(1,b_\beta),
\]
where $\otimes$ denotes the Kronecker product.

Finally, recall that a \emph{Vandermonde} matrix defined by $\boldnu=(\nu_1,\dots,\nu_n)\in\bF_q^n$ is a matrix whose $(i,j)$'th entry equals $\nu_j^{i-1}$. We say that a matrix~$V$ is a \emph{generalized Vandermonde} (GV) matrix defined by $\boldnu$ if~$V=M\cdot \diag(\boldd)$ for some Vandermonde matrix~$M$ defined by $\boldnu$ and some vector~$\boldd=(d_1,\ldots,d_n)$ with nonzero entries. Note that a GV matrix~$V\in \bF_{q}^{r\times s}$ for some integers~$s\ge r$, which is defined by $s$ distinct field elements, is also an MDS matrix, i.e., all its~$r\times r$ submatrices are invertible. 

\begin{theorem}\label{theorem:subfieldUniversal}
	For an integer~$\alpha=2^\beta$ and an integer~$n$, let~$q$ be a prime power such that~$q\ge{ n}$, and let~$V\in\bF_{q}^{\alpha\times n}$ be a Vandermonde matrix defined by distinct $n$ elements.  Then, for $\boldh=(h_1,h_2,\ldots,h_n)\triangleq\cH^{(\beta)}(V)$, the code~$\cC\triangleq \ker(\boldh)$ is a~$\cN_{\alpha\vert\bal}^n$-correcting code over the basis $\boldomega$ of~\eqref{equation:defrecbase}.
\end{theorem}

The proof of this theorem requires the following lemma.

\begin{lemma}\label{lemma:GeneralizedVandermonde}
	Let $\alpha=2^\beta$ and let~$V$ be an $\alpha\times n$ GV matrix defined by $\boldnu=(\nu_1,\dots,\nu_n)\in\bF_q^n$. Then for all~$0\le i\le \beta$, the matrix~$\cH^{(i)}(V)$ is a GV matrix over~$\bF_{q^{2^i}}$ also defined by $\boldnu$.
\end{lemma}

\begin{proof}
	We prove this claim by induction, in which the base case~$i=0$ is clear. Assume that~$V_i\triangleq \cH^{(i)}(V)\in\bF_{q^{2^i}}^{(\alpha/2^i)\times n}$ is a GV matrix, and let~$U_i$ and~$L_i$ be its upper and lower halves, respectively. Since~$V_i$ is a GV matrix, there exists a Vandermonde matrix~$M\in \bF_{q^{2^i}}^{(\alpha/2^i)\times n}$ defined by $\boldnu$ and a vector~$\boldd\in (\bF_{q^{2^i}}^*)^n$ such that~$V_i=M\diag(\boldd)$. Hence, it follows that~$U_i=\UH(M)\diag(\boldd)$ and~$L_i=\LH(M)\diag(\boldd)$, and therefore
	\begin{align*}
	V_{i+i}&=\cH^{(i+1)}(V)=\cH_{\beta-i}(V_i)\\
	&=U_i+b_{\beta-i}L_i\\
	&=\UH(M)\diag(\boldd)+b_{\beta-i}\LH(M)\diag(\boldd).
	\end{align*}
	Now, since~$M$ is a Vandermonde matrix, it is readily verified that $\LH(M)=\UH(M)\diag(\boldx)$ for some $\boldx=(x_1,\ldots,x_n)\in(\bF_{q^{2^i}}^*)^n$, and thus
	\begin{align*}
	V_{i+i}&=\UH(M)\diag(\boldd)+b_{\beta-i}\UH(M)\diag(\boldx)\diag(\boldd)\\
	&=\UH(M)\left(\diag(\boldd)+b_{\beta-i}\diag(\boldx)\diag(\boldd) \right)\\
	&=\UH(M)\diag((\1+b_{\beta-i}\boldx)\odot \diag(\boldd)),
	\end{align*}
	where~$\odot$ denotes the pointwise product of vectors (also called the \emph{Hadamard product}), and~$\1$ is the all~$1$'s vector. Since~$\UH(M)$ is a Vandermonde matrix defined by $\boldnu$, to finish the proof it suffices to show that the entries of~$(\1+b_{\beta-i}\boldx)\odot \diag(\boldd)$ are nonzero. Assuming otherwise, it follows that~$(1+b_{\beta-i}x_j)d_j=0$ for some~$j\in[n]$; and since~$d_j\ne 0$ and~$x_j\ne 0$, we have that~$b_{\beta-i}=-x_j^{-1}$. However,~$-x_j^{-1}\in\bF_{q^{2^i}}$ and~$b_{\beta-i}\notin\bF_{q^{\frac{\alpha}{2^{\beta-i}}}}=\bF_{q^{2^i}}$, a contradiction.
\end{proof}

We are now ready to prove Theorem~\ref{theorem:subfieldUniversal}.

\begin{proof}
	(of Theorem~\ref{theorem:subfieldUniversal}) 
	According to Lemma~\ref{lemma:linearIndependence}, it suffices to show that for any~$1\le i\le \min\{ \log n,\beta \}$, any~$2^i$-subset of~$\{h_j\}_{j\in [n]}$ is linearly independent over~$\bF_{q^{\alpha/2^i}}$. For any such~$i$, let~$J\subseteq [n]$ be a subset of size~$2^i$, and let~$H_J\in\bF_{q^{\alpha/2^i}}^{2^i\times 2^i}$ be the matrix whose columns are the representations of all elements in~$\{ h_j \}_{j\in J}$ over the (ordered) basis~$W_{i}$. Notice that~$\{ h_j \}_{j\in J}$ is a linearly independent set over~$\bF_{q^{\alpha/2^i}}$ if and only if~$H_J$ is invertible. However,~$H_J$ is a~$2^i\times 2^i$ submatrix of~$\cH^{(\beta-i)}(V)\in\bF_{q^{\alpha/2^i}}^{2^i\times n}$, which is a GV matrix defined by distinct elements according to Lemma~\ref{lemma:GV}, and hence also an MDS matrix. Thus,~$H_J$ is invertible, and the claim follows.
\end{proof}

\begin{remark}
	According to Theorem~\ref{theorem:subfieldUniversal} it follows that
	\begin{align*}
			h_j=\prod_{i=1}^{\beta}\left( 1+b_ia_j^{\alpha/2^i} \right)\mbox{ for all }j\in[n],
	\end{align*}
	where~$a_1,\ldots,a_n$ are the distinct $\bF_q$-elements in the underlying Vandermonde matrix~$V$.
\end{remark}

\begin{remark}
    The above construction is closely related to a classical coding theoretic notion called \emph{alternant codes}~\cite[Sec.~5.5]{Rot06}. An~$[n,k]_q$ \emph{Generalized Reed-Solomon} (GRS) code is a linear code whose parity check matrix is an~$(n-k)\times n$ GV matrix over~$\bF_q$. An alternant code~$\cC_{\mathrm{alt}}$ is defined as~$\cC\cap F^n$, where~$\cC$ is an~$[n,k]_q$ GRS code and~$F$ is a subfield of~$\bF_q$. 
Let~$\alpha<n$, and for any~$0\le i\le \beta$ let~$\cC_{i}$ be the right kernel of~$\cH^{(i)}(V)$ over~$\bF_{q^{2^i}}$. Notice that Lemma~\ref{lemma:GeneralizedVandermonde} shows that~$\cC_i$ is an~$[n,n-\alpha/2^i]_{q^{2^i}}$ GRS code. Furthermore, it is readily verified that~$\cC_j$ is an alternant code of~$\cC_i$ whenever~$j\le i$. Lemma~\ref{lemma:GeneralizedVandermonde} also implies that the codes we construct here have the property that all the alternant codes in the hierarchy are of maximum distance, and in cases where~$q$ is prime, these are all possible alternant codes.
\end{remark}

\subsection{Correcting power erasure patterns}\label{subsection:MoshesSubsection}

We generalize the results of the previous section by considering a larger family of erasure patterns, $\cN_{\alpha\vert\pow}^n$, that includes balanced patterns, i.e., $\cN_{\alpha\vert\bal}^n\subseteq \cN_{\alpha\vert\pow}^n$. As before, let $\alpha=2^\beta$ for some positive integer $\beta$. An erasure pattern $\boldt\in\cN_\alpha^n$ is called a \emph{power erasure pattern} if
there exists $J\subseteq[n]$ such that
\[t_j=\begin{cases}
\frac{\alpha}{2^{m_j}} & j\in J,\\
0 & \text{otherwise,}
\end{cases}
\]
where $0\leq m_j\leq \beta$ is an integer for all $j\in J$, and $\sum_{j\in J}2^{-m_j}=1$. Thus, for example, when $n=4$, $(\alpha/2,\alpha/4,\alpha/4,0)$ is a power erasure pattern but is not a balanced erasure pattern.

\begin{theorem}\label{theorem:subfieldUniversal1}
  For an integer $\alpha=2^\beta$, and an integer $n$, let $q$ be a prime power such that $\frac{\alpha}{2}|q-1$.  Let $\nu_1,\dots,\nu_n\in\bF_q$ be arbitrary non-zero scalars such that $\nu_j^{\alpha/2}\neq \nu_k^{\alpha/2}$ for all $j\neq k$.  Let $V\in\bF_{q}^{\alpha\times n}$ be a Vandermonde matrix defined by $(\nu_1,\dots,\nu_n)$. Then, for $\boldh=(h_1,h_2,\ldots,h_n)\triangleq\cH^{(\beta)}(V)$, the code~$\cC\triangleq\ker(\boldh)$ is an~$\cN_{\alpha\vert\pow}^n$-correcting code over the basis $\boldomega$ of~\eqref{equation:defrecbase}.
\end{theorem}

We shall require the following natural extension of Lemma~\ref{lemma:linearIndependence}.

\begin{lemma}\label{lemma:linearIndependence1}
  The code~$\cC$ of~\eqref{equation:cdef}  is~$\cN_{\alpha\vert\pow}^n$-correcting over a recursive basis~$\boldomega$ if and only if for every power erasure pattern $\boldt\in \cN_{\alpha\vert\pow}^n$ (defined by the sets $J$ and $\{m_j\}_{j\in J}$) the equation
  \[ \sum_{j\in J} h_j c_j = 0,\]
  has only the trivial solution when $c_j\in \bF_{q^{\alpha/2^{m_j}}}$ for every~$j\in J$.
\end{lemma}
\begin{proof}
  If~$\cC$ is not~$\cN_{\alpha\vert\pow}^n$-correcting, then there exists a nonzero $\boldc=(c_1,c_2,\ldots,c_n)$ in~$\cC$ and a power erasure pattern~$\boldt\in \cN_{\alpha\vert\pow}^n$ such that~$\boldc\in\cC\cap X_\boldt$. By the definition of~$\cN_{\alpha\vert\pow}^n$, it follows that there exist corresponding sets $J$ and $\{m_j\}_{j\in J}$. Hence, we have that
  \begin{align*}
    c_j\in\Span{\omega_1,\ldots,\omega_{\alpha/2^{m_j}}}=\bF_{q^{\alpha/2^{m_j}}}\mbox{ for all }j\in J,
  \end{align*}
  as well as $\sum_{j\in J}h_j c_j=0$, thus proving one direction of
  the claim. The proof of the inverse direction is similar.
\end{proof}

We now give the proof of Theorem~\ref{theorem:subfieldUniversal1}.

\begin{proof}
  (of Theorem~\ref{theorem:subfieldUniversal1})
  Let $\boldt\in \cN_{\alpha\vert\pow}^n$ be a power erasure pattern, with corresponding sets $J$ and $\{m_j\}_{j\in J}$. By applying Lemma~\ref{lemma:linearIndependence1} our goal is now to prove a solution to $\sum_{j\in J} h_j c_j=0$ with $c_j\in\bF_{q^{\alpha/2^{m_j}}}$ must be a trivial all-zero solution.

  Let us denote by $\boldv_j^\intercal$, $j\in[n]$, the $j$th column of the Vandermonde matrix~$V$. Additionally, recall the recursive basis $\boldomega\triangleq W_\beta$ from~\eqref{equation:defrecbase}. Thus, $\boldv_j^\intercal$ contains the coordinates (over $\bF_q$) of $h_j$ when using the basis $\boldomega$.
  
  If we define $\ov_j\triangleq(1,\nu_j,\dots,\nu_j^{\alpha/2^{m_j}-1})$ then
  \[ \boldv_j^\intercal=\begin{pmatrix}
  \ov_j^\intercal \\
  \nu_j^{\alpha/2^{m_j}}\ov^\intercal_j \\
  \vdots\\
  \nu_j^{(2^{m_j}-1)\alpha/2^{m_j}}\ov^\intercal_j
  \end{pmatrix}.\]
  Similarly, we define 
  \[\oom_j\triangleq W_{\beta-m_j}=(1,b_{m_j+1})\otimes (1,b_{m_j+2})\otimes\cdots \otimes (1,b_\beta),\]
  which is the $\alpha/2^{m_j}$-prefix of $\boldomega$. By the construction of the recursive basis $\boldomega$ we have that $\oom_j$ is a basis for $\bF_{q^{\alpha/2^{m_j}}}$. We now notice that
  \[ \begin{pmatrix}
  \oom_j\cdot\ov_j^\intercal \\
  \nu_j^{\alpha/2^{m_j}}\oom_j\cdot\ov^\intercal_j \\
  \vdots\\
  \nu_j^{(2^{m_j}-1)\alpha/2^{m_j}}\oom_j\cdot\ov^\intercal_j
  \end{pmatrix},\]
  is the coordinate vector of $h_j$ when $\bF_{q^\alpha}$ is viewed as a vector space over $\bF_{q^{\alpha/2^{m_j}}}$ using the ordered basis \[\hatom_j\triangleq 
  (1,b_1)\otimes (1,b_2)\otimes\cdots\otimes(1,b_{m_j}).\]

  By rewriting $c_j=\sum_{i=1}^{\alpha/2^{m_j}} c_{j,i} \omega_i$, with $c_{j,i}\in\bF_q$, our goal is equivalent to proving the set $\bigcup_{j\in J} \{ h_j\omega_1, \dots, h_j\omega_{\alpha/2^{m_j}}\}$ is linearly independent over $\bF_q$. For each $j\in J$, and for each $i\in[\alpha/2^{m_j}]$, we may write a column vector of the coordinates of $h_j\omega_i$ in $\bF_{q^{\alpha/2^{m_j}}}$ using the basis $\hatom$ as
  \[ \begin{pmatrix}
  \omega_i\oom_j\cdot\ov_j^\intercal \\
  \nu_j^{\alpha/2^{m_j}}\omega_i\oom_j\cdot\ov^\intercal_j \\
  \vdots\\
  \nu_j^{(2^{m_j}-1)\alpha/2^{m_j}}\omega_i\oom_j\cdot\ov^\intercal_j
  \end{pmatrix},\]
  where we note that both $\omega_i$ and $\oom_j\cdot\ov^\intercal$ are in $\bF_{q^{\alpha/2^{m_j}}}$, and $\nu_j\in\bF_q$. Now, viewing $\bF_{q^{\alpha/2^{m_j}}}$ as a vector space over $\bF_q$ using the basis $\oom_j$, multiplication by $\omega_i$ may be represented as a multiplication of the coordinates by $C_{j,i}$, an $\alpha/2^{m_j}\times \alpha/2^{m_j}$ matrix over $\bF_q$ ($C_{i,j}$ can be made explicit using companion matrices, but this is immaterial to the rest of the proof). Thus, the coordinates of $h_j\omega_i$ over $\bF_q$ using the basis $\boldomega$ take on the simple form of
  \[
  \boldz_{j,i}^\intercal\triangleq
  \begin{pmatrix}
    C_{j,i} & & & \\
    & C_{j,i} & & \\
    & & \ddots &\\
    & & & C_{j,i}
  \end{pmatrix}
  \cdot \boldv_j^\intercal =
  \begin{pmatrix}
  C_{j,i} \ov_j^\intercal \\
  \nu_j^{\alpha/2^{m_j}} C_{j,i}\ov^\intercal_j \\
  \vdots\\
  \nu_j^{(2^{m_j}-1)\alpha/2^{m_j}} C_{j,i}\ov^\intercal_j
  \end{pmatrix}
  \]
  
  If we define the matrix $Z\in\bF_q^{\alpha\times\alpha}$ to have as its columns $\{\boldz_{j,i}^\intercal\}$, $j\in J$, $i\in[\alpha/2^{m_j}]$, then it now suffices to prove $\det(Z)\neq 0$. Our strategy now is, for each $j\in J$, to take the $\alpha/2^{m_j}$ columns $\{\boldz_{j,i}^\intercal\}_{i\in[\alpha/2^{m_j}]}$ and replace them by using invertible column operations. The overall resulting matrix, $Z'$ will be shown to have $\det(Z')\neq 0$, implying $\det(Z)\neq 0$.

  Fix any $j\in J$. Obviously the set $\{h_j\omega_i\}_{i\in[\alpha/2^{m_j}]}$ is linearly independent over $\bF_q$ since $\{ \omega_{i} \}_{i\in [\alpha/2^{m_j}]}$ is, and therefore also $\{\boldz_{j,i}^\intercal\}_{i\in[\alpha/2^{m_j}]}$. We now contend that this implies that the set $\{C_{j,i}\ov_j^\intercal\}_{i\in[\alpha/2^{m_j}]}$ is linearly independent over $\bF_q$. Assuming to the contrary it is not, there exist $c_1,\dots,c_{\alpha/2^{m_j}}\in\bF_q$, not all zero, such that $\sum_{i\in[\alpha/2^{m_j}]} c_i C_{j,i} \ov_j^\intercal=0$, but then $\sum_{i\in[\alpha/2^{m_j}]} c_i \nu_j^{\ell/2^{m_j}}C_{j,i} \ov_j^\intercal=0$ for any integer $\ell$, implying $\sum_{i\in[\alpha/2^{m_j}]} \boldz_{j,i}^\intercal=0$, a contradiction.
  
  Let $\xi_j\in\bF_q$ be an element of multiplicative order $o(\xi_j)=\alpha/2^{m_j}$, the existence of which is guaranteed by the requirement $\frac{\alpha}{2}|q-1$. Since we established that $\{C_{j,i}\ov_j^\intercal\}_{i\in[\alpha/2^{m_j}]}$ is linearly independent over $\bF_q$, by invertible column operations we may map
  \begin{multline*}
  \left( C_{j,1}\ov_j^\intercal \left\vert
  C_{j,2}\ov_j^\intercal \left\vert
  \dots \left\vert
  C_{j,\alpha/2^{m_j}}\ov_j^\intercal
  \right.\right.\right.
  \right) \\
  \xmapsto{\hspace{2em}}
  \begin{pmatrix}
  1       & 1          & \dots & 1 \\
  \nu_j   & \xi_j\nu_j     & \dots & \xi_j^{\alpha/2^{m_j}-1}\nu_j\\
  \nu_j^2 & (\xi_j\nu_j)^2 & \dots & (\xi_j^{\alpha/2^{m_j}-1}\nu_j)^2\\
  \vdots  & \vdots         & \ddots& \vdots \\
  \nu_j^{\alpha/2^{m_j}-1} & (\xi_j\nu_j)^{\alpha/2^{m_j}-1} & \dots & (\xi_j^{\alpha/2^{m_j}-1}\nu_j)^{\alpha/2^{m_j}-1}\\
  \end{pmatrix},
  \end{multline*}
  i.e., the square Vandermonde matrix defined by $(\nu_j,\xi_j\nu_j,\xi_j^2\nu_j,\dots,\xi_j^{\alpha/2^{m_j}-1}\nu_j)$, which we denote by $V_j$ for convenience. Using the same column operations on $\{\boldz_{j,i}^\intercal\}_{i\in[\alpha/2^{m_j}]}$ the mapping becomes
  \begin{align*}
  \left( \boldz_{j,1}^\intercal \left\vert
  \boldz_{j,2}^\intercal \left\vert
  \dots \left\vert
  \boldz_{j,\alpha/2^{m_j}}^\intercal
  \right.\right.\right.
  \right) &\mapsto
  \begin{pmatrix}
  V_j \\
  \nu_j^{\alpha/2^{m_j}}V_j \\
  \vdots\\
  \nu_j^{(2^{m_j}-1)\alpha/2^{m_j}}V_j
  \end{pmatrix}\\
  &=
  \begin{pmatrix}
  1       & 1          & \dots & 1 \\
  \nu_j   & \xi_j\nu_j     & \dots & \xi_j^{\alpha/2^{m_j}-1}\nu_j\\
  \nu_j^2 & (\xi_j\nu_j)^2 & \dots & (\xi_j^{\alpha/2^{m_j}-1}\nu_j)^2\\
  \vdots  & \vdots         & \ddots& \vdots \\
  \nu_j^{\alpha-1} & (\xi_j\nu_j)^{\alpha-1} & \dots & (\xi_j^{\alpha/2^{m_j}-1}\nu_j)^{\alpha-1}\\
  \end{pmatrix},
  \end{align*}
  which is an $\alpha\times (\alpha/2^{m_j})$ Vandermonde matrix. 
  
  We repeat the above process for each $j\in J$ to obtain the matrix $Z'$ which satisfies $\det(Z')=\xi\det(Z)$ for some
  $\xi\in\bF_q$, $\xi\neq 0$, since only invertible column operations were used. Finally, we note that $Z'$ is itself a
  Vandermonde matrix that is defined by (the multiset) $\bigcup_{j\in J}\{\xi_j^{i-1}\nu_j\}_{i\in[\alpha/2^{m_j}]}$ (in some order), and since $\nu_j^{\alpha/2}\neq
  \nu_k^{\alpha/2}$ for all $j\neq k$, we have $\det(Z')\neq 0$, as
  desired.
\end{proof}

As a final note, we observe the field size requirements imposed by Theorem~\ref{theorem:subfieldUniversal1}. We need to choose $n$ distinct non-zero values from $\bF_q$. However, each choice precludes some other elements from being chosen. More specifically, let $\xi\in\bF_q$ be an element with multiplicative order $\frac{\alpha}{2}$, and let $\Span{\xi}$ be the multiplicative group spanned by it. Then we may choose at most one element from each of the cosets in $\bF_q^*/\Span{\xi}$. Hence, $q\geq \frac{\alpha}{2}n+1$.

\subsection{Correcting bounded erasure patterns} \label{subsection:Gabidulin}
In this subsection it is shown that Gabidulin codes, a well-known family of rank-metric codes, are capable of protecting against a large family of erasure patterns. In particular, for~$\alpha\ge n$ and an integer~$r\le n$, the code~$\Gab[n,n-r]_{q^\alpha}$, defined below, can protect against~$\cT_r\triangleq \cN_{r,nr}^n=\{0,1,\ldots,r\}^n$. Notice that~$\cT_r$ does not include full erasures of codeword symbols (unless the code is trivial), and yet Gabidulin codes can protect against erasures in the usual sense (see~\cite{richter2004fast}).


For the next theorem, recall that a linearized polynomial is a polynomial over~$\bF_{q^\alpha}$ in which all nonzero coefficients correspond to monomials of the form $x^{q^i}$ for some nonnegative integer~$i$. For a linearized polynomial~$f$, let its~$q$-degree be~$\deg_q(f)\triangleq \log_q(\deg f)$. It is widely known that any function from~$\bF_{q^\alpha}$ to itself, which is linear over~$\bF_q$, corresponds to a linearized polynomial. The following theorem applies over any basis~$\boldomega$.

\begin{theorem}\label{theorem:Gab}
	For nonnegative integers~$r,n$, and~$\alpha$ such that~$n\le \alpha$ and~$r<n$, the code
	\begin{align*}
	\Gab[n,n-r]_{q^\alpha}\triangleq \set*{( f(\omega_1),\ldots,f(\omega_n) ) ; \text{$f$ is linearized and $\deg_q (f)<n-r$} }
	\end{align*} 
	is~$\cT_r$-correcting.
\end{theorem}

\begin{proof}
	We show that~$\Gab[n,n-r]\cap \cX_\boldt=\{0\}$ for all~$\boldt\in \cT_r$. Assuming otherwise, we have a pattern~$\boldt\in \cT_r$ and a nonzero linearized polynomial~$f$ of~$q$-degree less than~$n-r$ such that
	\begin{align}\label{equation:f}
	f(\omega_j) &\in \Span{\omega_1,\ldots,\omega_{t_j}} ,\mbox{ for all }j\in[n].
	\end{align}
	Since~$f$ is a linearized polynomial and since~$\boldt\in \cT_r$, Eq.~\eqref{equation:f} implies that $f(\Span{\omega_1,\ldots,\omega_n} )\subseteq \Span{\omega_1,\ldots,\omega_{r}} $, which in turn implies that~${\dim \ker(f)\ge n-r}$. Thus,~$f$ has more roots than its degree, which is a contradiction.
\end{proof}

Note that~$n\le \alpha$ is necessary, since the evaluation points~$\omega_1,\ldots,\omega_n$ must be linearly independent over~$\bF_q$. 
Finally, we emphasize that this construction applies to any~$q$.

\section{Lower Bound}\label{section:bounds}

First, it is clear that any~$m$-correcting code~$\cC\subseteq \bF_{q^\alpha}^n$ can correct~$m'\triangleq\floor{m/\alpha}$ erasures in the usual sense. Therefore, the well-known Singleton bound implies that $m'\le n-k $. Moreover, in cases where~$m'=n-k$, namely, when~$\cC$ is an MDS code, the MDS conjecture (e.g., see \cite{MacSlo78}, and its resolution in certain cases \cite{Bal12,BalDeB12}) implies~$q^\alpha\ge n-1$. In the remainder of this section a Gilbert-Varshamov type argument is used to prove the following existence theorem.
\begin{theorem}\label{theorem:GV}
	For all positive integers~$n,m,\alpha,$ and~$r$ such that~$m< \alpha(r-1)$, if
	\[q>\left((m+1)\binom{m+n-2}{ n-2}\right)^{\frac{1}{\alpha(r-1)-m}}\]
	then there exists an~$[n,n-r]_{q^\alpha}$~$m$-correcting code~$\cC$.
\end{theorem}

Before proving the theorem, we prove an auxiliary claim, which applies for any basis~$\boldomega$. We say that a matrix over~$\bF_{q^\alpha}$ is $m$-\emph{good} (good, in short) if its right kernel does not contain nonzero vectors~$\boldx$ with~$w(\boldx)\le m$. In the proof of Theorem~\ref{theorem:GV} we choose the columns of the parity-check matrix of the code one after another, while showing that there always exists an eligible column to add; the question of column eligibility boils down to the following lemma.

\begin{lemma}\label{lemma:GV}
	If~$H_\ell\triangleq (\boldg_1^\intercal,\ldots,\boldg_\ell^\intercal)\in\bF_{q^\alpha}^{r\times \ell}$ is good and 
	\begin{align}\label{equation:Rl}
		\boldg_{\ell+1}^\intercal\notin \set*{ \gamma\cdot\sum_{i=1}^\ell x_i \boldg_i^\intercal ; \text{$\gamma\in\bF_{q^\alpha}$ and $w(x_1,\ldots,x_\ell)\le m$}}\triangleq R_\ell
	\end{align}
	then~$H_{\ell+1}\triangleq (\boldg_1^\intercal,\ldots,\boldg_{\ell+1}^\intercal)\in\bF_{q^\alpha}^{r\times (\ell+1)}$ is good.
\end{lemma}
\begin{proof}
	Assume to the contrary that the right kernel of~$H_{\ell+1}$ contains a nonzero vector~$\boldx=(x_1,\ldots,x_{\ell+1})\in\bF_{q^\alpha}^{\ell+1}$ such that~$w(\boldx)\le m$, which implies that $-x_{\ell+1}\boldg_{\ell+1}^\intercal=\sum_{i=1}^{\ell}x_i\boldg_i^\intercal$ and that~$w(x_1,\ldots,x_{\ell})\le m$. If~$x_{\ell+1}=0$, it follows that the vector~$\boldx'\triangleq(x_1,\ldots,x_\ell)$  satisfies~$H_\ell\boldx'=0$ and~$w(\boldx')\le m$, in contradiction to~$H_\ell$ being good. Otherwise, we have that~$\boldg_{\ell+1}^\intercal=(-x_{\ell+1}^{-1})\cdot\sum_{i=1}^\ell x_i\boldg_i^\intercal$, and hence~$\boldg_{\ell+1}^\intercal\in R_\ell$ in contradiction with~\eqref{equation:Rl}.
\end{proof}

The following two properties are easy to prove.

\begin{lemma}
\label{lemma:sizeRell}
For the sets $R_{\ell}$ from~\eqref{equation:Rl},
\begin{enumerate}
	\item $|R_{n-1}|\ge |R_{\ell}|$ for all~$\ell\le n-1$.
	\item $|R_{n-1}|\le q^{\alpha}\sum_{i=0}^{m}q^i\binom{i+n-2}{ n-2}\leq (m+1)q^{\alpha+m}\binom{m+n-2}{n-2}$.
\end{enumerate}
\end{lemma}
\begin{proof}
The first property is due to simple monotonicity. For the second property we upper bound the size of the set by assuming that all the linear combinations in the definition of the set are distinct. Then, we have $q^{\alpha}$ ways of choosing $\gamma$. Finally, the number of vectors $\boldx\in\bF_{q^\alpha}^{n-1}$ with $w(\boldx)\leq m$ may be found using a standard balls-into-bins argument to be $\sum_{i=0}^m q^i\binom{i+n-2}{n-2}$. Since $q^i\binom{i+n-2}{n-2}$ is increasing in $i$ we obtain the final inequality.
\end{proof}

We are now in a position to prove Theorem~\ref{theorem:GV}.

\begin{proof}
	(of Theorem~\ref{theorem:GV}) We construct the parity check matrix of the code~$\cC$ column by column, starting from an~$r\times r$ identity matrix. Clearly, it suffices to guarantee that all along this construction, the resulting matrices are good; this would guarantee that~$\cC\cap \cX_\boldt=\{0\}$ for every~$\boldt\in\cN_{\alpha,m}^n$, and thus that~$\cC$ is $m$-correcting by Lemma~\ref{lemma:main}. 
	
	Assume that~$H_\ell\in\bF_{q^\alpha}^{r\times \ell}$ is good for some~$\ell\ge r$ (for~$\ell=r$ the goodness is satisfied since there are no nonzero vectors in the kernel). According to Lemma~\ref{lemma:GV} and the above observations, it follows that if~$|\bF_{q^\alpha}^r|-|R_{n-1}|>0$, then there exists a legitimate choice for the added column~$\boldg_{\ell+1}^\intercal$. Hence, by the bound on~$|R_{n-1}|$ from Lemma~\ref{lemma:sizeRell} we have
	\begin{align*}
		|\bF_{q^\alpha}^r|-|R_{n-1}|\ge q^{\alpha r}-(m+1)q^{\alpha+m}\binom{m+n-2}{ n-2}.
	\end{align*}
	If that is strictly larger than zero, the desired code exists. Thus, it suffices to require
	\begin{align*}
		q^{\alpha r-\alpha-m}&>(m+1)\binom{m+n-2}{ n-2}\\
		q&>\left((m+1)\binom{m+n-2}{ n-2}\right)^{\frac{1}{\alpha(r-1)-m}}.\qedhere
	\end{align*}
\end{proof}
In the remainder of this section the bound on~$q$ in Theorem~\ref{theorem:GV} is analyzed asymptotically in the two regimes of interest (see Subsection~\ref{subsection:applications}). In both regimes we focus on the practically important case where the dimension~$k$ (and hence~$r$) is proportional to~$n$, and the erasure correction capability~$m$ is proportional to~$\alpha n$; this corresponds to erasure correction of a constant fraction of the information symbols.

In the case~$\alpha\gg n$ the parameter~$n$ is seen as constant and the parameter~$\alpha$ tends to infinity. Say that $m=c_1\alpha$ and~$\alpha(r-1)-m=c_2\alpha$ for some constants~$c_1,c_2$, and then the condition on $q$ from Theorem~\ref{theorem:GV} becomes
\begin{align*}
	q>\left((c_1\alpha+1)\binom{c_1\alpha+n-2}{ n-2}\right)^{\frac{1}{c_2\alpha}}=\poly(\alpha)^{\frac{1}{\Theta(\alpha)}}\xrightarrow{\alpha\to\infty}1.
\end{align*}

In the case~$n\gg \alpha$ we view~$\alpha$ as constant and~$n$ as tending to infinity. Say that~$m=c_1n$ and~$\alpha(r-1)-m=c_2n$ for some~$c_1,c_2$. By the well known approximation of the binomial coefficient (e.g., see \cite[Lemma 7, p.~309]{MacSlo78}), the condition on $q$ from Theorem~\ref{theorem:GV} becomes 
\begin{align*}
    q&> \left((c_1 n+1)\binom{(1+c_1)n-2}{n-2}\right)^{\frac{1}{c_2n}}\\
    &= \left(2^{(1+c_1)nH\left(\frac{1}{1+c_1}\right)(1+o(1))}\right)^{\frac{1}{c_2n}}\xrightarrow{n\to\infty}2^{\frac{1+c_1}{c_2}H\left(\frac{1}{1+c_1}\right)},
\end{align*}
where~$H(x)\triangleq -x\log_2(x)-(1-x)\log_2(1-x)$ is the binary entropy function.

\section*{Acknowledgments}
Moshe Schwartz was supported in part by a German Israeli Project Cooperation (DIP) grant under grant no. PE2398/1-1. Yuval Cassuto was supported in part by a grant from the US-Israel Binational Science Foundation.



\bibliographystyle{elsarticle-num}
\bibliography{allbib}

\appendix

\section{\texorpdfstring{$\alpha$}{alpha}-correcting codes from mutual eigenvector of UDMs}\label{appendix:mutualEigenvector}
For the case~$m=\alpha$, there exists an intriguing connection between UDMs and~$\alpha$-correcting codes.

\begin{theorem}\label{theorem:equivalence}
	For~$h_1,\ldots,h_n\in\bF_{q^\alpha}$, a code~$\cC=\set*{ \boldc\in\bF_{q^\alpha}^n ; (h_1,\ldots,h_n)\cdot \boldc^\intercal=0 }$ is an~$\alpha$-correcting code over an ordered basis~$\boldomega\triangleq (\omega_1,\ldots,\omega_\alpha)$ if and only if there exists a set~$A_1,\ldots,A_n$ of UDMs over~$\bF_q$ such that for any~$i\in[n]$, the element~$h_i$ is an eigenvalue of~$A_i$ with a corresponding eigenvector~$\boldomega^\intercal$.
\end{theorem}
\begin{proof}
	Let~$A_1,\ldots,A_n\in\bF_q^{\alpha\times \alpha}$ be UDMs with eigenvalues~$h_1,\ldots,h_n\in\bF_{q^\alpha}$, respectively, all of which correspond to the eigenvector~$\boldomega$, i.e.,
	\begin{align}\label{equation:eigenvectors}
	A_i\boldomega^\intercal	=h_i\boldomega^\intercal\mbox{ for all }i\in[n].
	\end{align}
	If~$\cC$ is not~$\alpha$-correcting, it follows that there exist~$\boldt\in\cN_\alpha^n$ and a nonzero codeword~$\boldc=(c_1,c_2,\ldots,c_n)\in \cC$ such that~$c_i\in\Span{\omega_1,\ldots,\omega_{t_i}} $ for all~$i\in [n]$, and therefore
	\begin{align*}
	h_ic_i \in \Span{h_i\omega_1,\ldots,h_i\omega_{t_i}}
	\overset{\eqref{equation:eigenvectors}}{=}\Span{A_i^{(1)}\boldomega^\intercal,\ldots,A_i^{(t_i)}\boldomega^\intercal} ,
	\end{align*}
	where~$A_i^{(j)}$ denotes the~$j$-th row of~$A_i$. In turn, this implies that for all~$i\in[n]$ there exists a nonzero vector~$\boldv_i\in\bF_{q}^{t_i}$ such that~$\boldv_iA_i^{(1:t_i)}\boldomega^\intercal=h_i c_i$, where for any positive integers~$r$ and~$s$, the notation~$A_i^{(s:r)}$ stands for the submatrix of~$A_i$ which consists of rows~$s$ through~$r$. Thus, 
	we have a nonzero vector~$\boldv\triangleq (\boldv_1\vert \boldv_2\vert\ldots \vert \boldv_n)\in\bF_q^\alpha$ that satisfies
	\begin{align}\label{equation:UDMproof}
	\boldv\cdot 
	\begin{pmatrix}
	A_1^{(1:t_1)}\\
	A_2^{(1:t_2)}\\
	\vdots\\
	A_n^{(1:t_n)}\\
	\end{pmatrix}\cdot \boldomega^\intercal
	=\sum_{i\in[n]}\boldv_iA_i^{(1:t_i)}\boldomega^\intercal=\sum_{i\in[n]}h_i c_i=0.
	\end{align}
	Now, since the entries of~$\boldomega$ are a basis, and since the~$A_i$'s and the~$\boldv_i$'s are over~$\bF_q$, the expression~$(\sum_{i\in[n]}\boldv_i A_i^{(1:t_1)})\boldomega^\intercal=0$ implies that the vector~$\sum_{i\in[n]}\boldv_i A_i^{(1:t_1)}$ is the zero vector. However, this implies that there exists a nonzero vector~$\boldv$ in the left kernel of a matrix which consists of upper rows of UDMs, a contradiction.
	
	Conversely, assume that~$\cC$ is~$\alpha$-correcting, and define matrices~$A_1,\ldots,A_n\in\bF_q^{\alpha\times \alpha}$ as follows. For every~$i\in[n]$, let~$A_i$ be the matrix  such that~$A_i^{(j)}$ is the expansion of~$h_i\omega_j$ over the basis~$\boldomega $, i.e.,~$h_i\omega_j= \sum_{\ell=1}^\alpha (A_{i}^{(j)})_{\ell}\omega_\ell$. Assuming to the contrary that~$A_1,\ldots,A_n$ are not UDMs, we have an element~$\boldt=(t_1,\ldots,t_n)\in\cN_{\alpha}^n$ and a nonzero vector~$\boldv\in\bF_{q}^\alpha$ such that
	\begin{align*}
	\boldv\cdot 
	\begin{pmatrix}
	A_1^{(1:t_1)}\\
	A_2^{(1:t_2)}\\
	\vdots\\
	A_n^{(1:t_n)}\\
	\end{pmatrix}=0.
	\end{align*}
	Partition~$\boldv$ to~$n$ consecutive parts~$\boldv_1,\boldv_2,\ldots,\boldv_n$ of sizes~$t_1,\ldots,t_n$, respectively, let~$c_i\triangleq \boldv_i\cdot (\omega_1,\ldots,\omega_{t_i})^\intercal$ for all~$i\in[n]$, and let~$\boldc\triangleq(c_1,\ldots,c_n)$. Notice that~$\boldc\in\cC$, since:
	\begin{align*}
	    (h_1,\ldots,h_n)\boldc^\intercal&=\sum_{i=1}^n h_i \boldv_i(\omega_1,\ldots,\omega_{t_i})^\intercal=\sum_{i=1}^n  \boldv_i(h_i\omega_1,\ldots,h_i\omega_{t_i})^\intercal\\
	    &= \sum_{i=1}^n \boldv_i \left( \sum_{\ell=1}^\alpha (A_i^{(1)})_\ell\omega_\ell,\ldots, \sum_{\ell=1}^\alpha (A_{i}^{(t_i)})_\ell \omega_\ell\right)^\intercal\\
	    &=\sum_{i=1}^n \boldv_i A_i^{(1:t_i)}\boldomega^\intercal\\
	    &= \boldv\cdot \begin{pmatrix}
	A_1^{(1:t_1)}\\
	A_2^{(1:t_2)}\\
	\vdots\\
	A_n^{(1:t_n)}\\
	\end{pmatrix}\cdot\boldomega^\intercal=0.\end{align*}
	Moreover, since $\boldc\in \cX$ by definition, it follows that~$\boldc$ is a nonzero codeword in~$\cC\cap \cX_\boldt$, a contradiction to~$\cC$ being an~$\alpha$-correcting code. 
\end{proof}

Finally, we note that Theorem~\ref{theorem:n=2} can alternatively be proved by a direct application of Theorem~\ref{theorem:equivalence}, and the details are left to the curious reader.

\section{An omitted proof}\label{appendix:n=2Extension}
\begin{proof}
    (of Corollary~\ref{corollary:n=2Extension}). First, we ought to show that such UDMs exist. Indeed, according to~\cite[Lemma~4]{VonGan06}, it follows that for any UDMs~$\{B_i\}_{i=1}^n$ and any lower-triangular invertible matrices~$\{C_i\}_{i=1}^n$, the matrices~$\{A_i=C_iB_i\}_{i=1}^n$ are UDMs as well. The existence of suitable UDMs for our proof is then proved by letting~$\{A_i\}_{i=1}^n$ be, say, the UDMs from Theorem~\ref{theorem:Vontobel} for the parameters at hand, letting~$C_i$ be an identity matrix for every~$i\in[n]\setminus\{2\}$, and
    \begin{align*}
        C_2=
        \begin{pmatrix}
            1&~&~&~&~&~\\
            ~&\ddots&~&~&~&~\\
            ~&~&1&~&~&~\\
            ~&~&-a_1&-a_0&~&~\\
            ~&\iddots&~&~&\ddots&~\\
            -a_1&~&~&~&~&-a_0\\
        \end{pmatrix}.
    \end{align*}
    Now, observe that~$\boldmu^\intercal$ is an eigenvector for the eigenvalue~$1$ of~$A_1$, and an eigenvector for the eigenvalue~$b$ of~$A_2$ (see~\ref{appendix:mutualEigenvector} for further implications of such mutual eigenvectors). Therefore, the square parity check matrix $(A_1^\intercal\boldmu^\intercal\vert\cdots\vert A_n^\intercal\boldmu^\intercal)$ has at least two dependent columns, which implies that $\dim\cC\ge 1$.
\end{proof}

\end{document}